\documentclass[10pt]{article}
\usepackage{amssymb}
\usepackage{amsmath}
\usepackage[all]{xy}
\usepackage{amsmath}
\usepackage{amsfonts}
\usepackage{amssymb}
\usepackage{amscd}
\usepackage{amsthm}
\usepackage{latexsym}
\usepackage{amsbsy}
\usepackage{graphicx}
\usepackage{float}
\usepackage[toc,page]{appendix}
\usepackage{caption}
\usepackage{subcaption}
\usepackage[utf8]{inputenc}
\usepackage{subcaption}

\usepackage{hyperref}
\hypersetup{
	colorlinks,
	citecolor=black,
	filecolor=black,
	linkcolor=black,
	urlcolor=black
}
\makeatletter
\theoremstyle{definition}

\newtheorem{theorem}{Theorem}[section]

\newtheorem{lemma}[theorem]{Lemma}

\DeclareMathOperator{\tr}{Tr}

\usepackage[most]{tcolorbox}
\newtcolorbox{highlighted}{colback=yellow,coltext=black,breakable}

\usepackage{soul}

\def\section{\@startsection{section}{1}{\z@}{-3.25ex plus -1ex minus
		-.2ex}{1.5ex plus .2ex}{\normalfont\bfseries}}
\def\subsection{\@startsection{subsection}{1}{\z@}{-3.25ex plus -1ex
		minus -.2ex}{1.5ex plus .2ex}{\normalfont\itshape}}
\date{}
\makeatother

\title{Coloured combinatorial maps and quartic bi-tracial 2-matrix ensembles from noncommutative geometry}
\author{ Masoud Khalkhali and  Nathan Pagliaroli\\
	Department of Mathematics, University of Western Ontario\\
	London, Ontario, Canada\footnote{\emph{Email addresses}:  masoud@uwo.ca, npagliar@uwo.ca}}

\begin{document}
	\maketitle
	\begin{abstract}
		We compute the first twenty moments of three  convergent quartic bi-tracial 2-matrix ensembles in the large $N$ limit.  These ensembles are toy models for Euclidean quantum gravity originally proposed by John Barrett and collaborators. A perturbative solution is found for the first twenty moments using the Schwinger-Dyson equations and properties of certain bi-colored unstable maps associated to the model. We then apply a result of Guionnet et al. to show that the perturbative and convergent solution coincide for a small neighbourhood of the coupling constants. For each model  we compute an explicit expression for the free energy, critical points, and critical exponents in the large $N$ limit. In particular, the string susceptibility is found to be $\gamma =1/2$, hinting that the associated universality class of the model is the continuous random tree.
	\end{abstract}
	\tableofcontents
	\section{Introduction}
	
	When attempting to define theories of Euclidean quantum gravity, one is usually interested in making sense of path integrals over some class of Riemannian metrics. In the context of noncommutative geometry, spectral triples are analogous to manifolds, and in some sense spectral triples generalize manifolds. In particular, for Riemannian spin$^{c}$ manifolds all the metric information can be recovered from the associated spectral triple via Connes' distance formula \cite{connes2013spectral}. With this in mind, Barrett \cite{barrett2015matrix} proposed defining path integrals over the moduli spaces of Dirac operators instead of those of metrics. To make these integrals well-defined, Barrett considered finite approximations of spectral triples called fuzzy  geometries. The resulting integrals are matrix integrals. The hope after fuzzifying these spectral triples is that in some limit one might be able to recover path integrals over metric spaces of Reimannian spectral triples. Thus, in some sense,  this would recover a theory of Euclidean quantum gravity. The limits of finite approximations of spectral triples is an active area of study \cite{rieffel2004gromov,rieffel2023dirac,connes2021spectral,latremoliere2023continuity}.
	
	While the work surrounding these models is not quite at this point of development, we have seen hints of the continuous theory in various limits. As originally pointed out in \cite{barrett2016monte}, in the large $N$ limit the spectral density function of the Dirac operators of certain Dirac ensembles bears resemblances to the spectra of Dirac operators on spin manifolds. This was explored quantitatively in \cite{barrett2019spectral}. More recently, in the double scaling limit, various Dirac ensembles have been shown to have the same critical exponents and satisfy the same differential equations as various minimal models from conformal field theory \cite{hessam2023double}. 
	One approach to this problem might be  to consider random metric spaces of maps defined by the perturbative expansion of these models. Based on the critical exponents found in \cite{hessam2023double}, these random metric spaces are expected to converge to the Brownian map \cite{le2013uniqueness}, but this will be explored in a future work. This is in contrast to the models studied in this paper, for which we find  strong hints that the associated random metric spaces converge to the continuous random tree.

	The models of interest in this article are three 2-matrix bi-tracial ensembles proposed by Barrett and Glaser \cite{barrett2016monte}. These models have been studied numerically in \cite{barrett2016monte,glaser2017scaling, barrett2019spectral,d2022numerical,glaser2023computational} via Monte Carlo simulations, which provided evidence of a spectral phase transition. Additionally, analytical bounds for the moments of these models were computed in \cite{hessam2022bootstrapping} using the bootstrap technique. We will show that all three ensembles have the same moments in the large $N$ limit and that one only needs to consider the following effective ensemble:
	$$\frac{1}{Z} e^{-S_{\text{eff}}(A,B)}dAdB,$$
	where this is a probability distribution on the space of pairs of $N$ by $N$ Hermitian matrices. The measure $dA dB$ is the product Lebesgue measure, and the potential is
	\begin{equation*}
		\begin{aligned}
			S_{\text{eff}}(A,B) = & 4 t_2 \left[N\operatorname{Tr}A^2+\tr B^2\right]  +4 t_{4} N\left[\operatorname{Tr} A^4 +\tr B^4+4 \operatorname{Tr} A^2 B^2-2 \operatorname{Tr} A B A B\right] \\
			& +12t_{4} \left[\left(\operatorname{Tr} A^2\right)^2+\left(\operatorname{Tr} B^2\right)^2\right] +8 t_{4} [\operatorname{Tr} A^2 \operatorname{Tr} B^2 ].
		\end{aligned}
	\end{equation*}
	The probability measure depends on two real coupling constants $t_{2}$ and $t_{4}$, where $t_{4}>0.$ 
	
	These models are interesting purely from the perspective of random matrix theory. Despite the success of single matrix ensembles \cite{mehta2004random,ercolani2003asymptotics}, in general very little is known about multi-matrix ensembles. For perturbative models,  much is known for general potentials in which the only unitarily invariant term is a $\tr AB$ interaction \cite{eynard1998matrices,chekhov2006free}. Besides this, only special cases of more general multi-matrix interactions are known \cite{kazakov2000solvable,eynard2011formal}. In convergent models even less is analytically tractable \cite{chadha1981method,guionnet2004first}. However, progress has been made in determining when the asymptotics of convergent models coincide with their perturbative counterparts \cite{guionnet2005combinatorial,bekerman2015transport}. Such results will be utilized in this work, allowing us to work with a perturbative expansion but make conclusions about convergent matrix integrals. An additional level of complexity in the matrix integrals proposed by Barrett comes from the fact that they are bi-tracial. As far back as the 90's physicists were already interested in studying integrals known as $\textit{multi-trace}$ in which the potential function $S(M)$ includes some product of traces.  Multi-tracial matrix integrals have appeared in many other areas of study \cite{Macroscopicandmicroscopic,Planar,Witten:2001ua,fuzzyfield}, and have seen the development of tools for both the perturbative and convergent cases \cite{borot2014formal,borot2017blobbed,de1995statistical}.

	Despite Barrett's models having highly nontrivial 2-matrix interactions in their potentials, in this article we derive  explicit formulae for the first twenty moments in terms of the coupling constants $t_{2}$ and $t_{4}$.  The idea of the derivation is to first consider the perturbative expansion of the models and study them as the generating functions of certain types of maps. Certain properties of these maps and the associated Schwinger-Dyson equations will allow us to deduce the moments. Applying a result of \cite{guionnet2005combinatorial} shows that these formulae are equivalent to the moments of the convergent ensembles in some neighbourhood of the coupling constants near zero. From these formulae for the first few moments we deduce the free energy, critical points, and critical exponents of the model. Note that partition functions of these models are the same whether they is written in terms of $D$ or $A$ and $B$, so even though we are working with moments in terms of $A$ and $B$, instead of $D$, we will be able to compute the above mentioned quantities of interest. Work towards explicit formulae for tracial powers of the Dirac operators of fuzzy geometries in terms of the tracial powers of constituent matrices can be found in \cite{perez2022computing}.  It is out hope that these results will lead to more analytic results as well as more direct approaches to studying these models in future works.

	In the following section we introduce the necessary background from Noncommutative Geometry and state the main results. In Section 2, we outline the derivation of the  Schwinger-Dyson equations for the model and its properties. Section 3 gives a short review of the relevant kind of maps before proving several useful properties that allow for the derivation of the moments. Section 4 shows the computation of the free energy as well as the critical exponents. Section 5 outlines the future work and implications of our results.
	
	\section{Background and Summary}
	\subsection{Random fuzzy geometries}
	In \cite{connes2013spectral}, Connes introduced the notation of a spectral triple $(\mathcal{A},\mathcal{H}, D)$ in which 
	\begin{itemize}
		\item $\mathcal{A}$ is a unital, involutive, complex, and associative algebra.
		\item The complex Hilbert space $\mathcal{H}$ is acted on by elements of  $\mathcal{A}$.
		\item   The Dirac operator $D$ is a self-adjoint operator acting on $\mathcal{H}$, that is in general unbounded. 
	\end{itemize} 
	These objects are additionally required to satisfy some regularity conditions. However, we are interested in spectral triples that automatically satisfy these conditions, so such details will be omitted. In particular, we are interested in real spectral triples, which have even more additional structures. The motivation to study real spectral triples is that they serve as noncommutative analogs of   $\text{spin}^c$  Riemannian manifolds. This idea is based on the fact that any closed $\text{spin}^c$ Riemannian manifold $M$ gives rise to a real spectral triple in which the algebra $\mathcal{A}= C^{\infty} (M)$ is the algebra of smooth complex valued functions on  $M$ and the Hilbert space is the space of square integrable sections of the spinor bundle such that the elements of $\mathcal{A}$ act as multiplication operators. The Dirac operator $D$ is the usual Dirac operator of $M$, and acts on the spinors. The additional structure mentioned before consists of standard charge conjugation and chirality operators, $J$ and $\gamma$. Conversely, the reconstruction theorem of Connnes tells us that under some natural conditions a real spectral triple with a commutative algebra can be realized as the real spectral triple of a $\text{spin}^c$ Riemannian manifold \cite{connes2013spectral}. 
	
	Fuzzy spaces have been studied as a method of regularization of commutative spaces since the fuzzy sphere in \cite{madore1992fuzzy}. In particular, they can be characterized within the formalism of spectral triples and are called fuzzy geometries or fuzzy spectral triples \cite{barrett2015matrix,barrett2019finite}. From a physics perspective these can  be thought of as $\text{spin}^c$  Riemannian manifolds with a finite resolution or Planck length. 
	
	A $(p,q)$ fuzzy geometry is a  real spectral triple of the form $(M_{N} (\mathbb{C}),V \otimes M_{N}(\mathbb{C}), D; J, \Gamma)$
	in which 
	\begin{itemize}
		\item The algebra of functions is replaced by the algebra of $N$ by $N$ complex matrices.
		\item The Hilbert space is some Hermitian irreducible Clifford module of signature $(p,q) $ with the charge conjugation operator $J$ and grading $\Gamma$  when the KO dimension $p+q$ is even.
		\item $D$ is a self-adjoint matrix that satisfies the so-called zero order and first order conditions \cite{barrett2015matrix}.
	\end{itemize}

	A result from \cite{barrett2015matrix} is that all Dirac operators satisfying the above-mentioned conditions can be expressed as
	\begin{equation}
		D=\sum_{I} \gamma^{I}\otimes \{K_{I},\cdot\}_{e_{I}}
		\label{general fuzzy Dirac},
	\end{equation}
	where the sum is over index sets of the form $\{i_{1} \leq ...\leq i_{k}\}$ with each index between one and $p+q$. The operator $\gamma^{I}$ denotes some product of gamma matrices. 
	If $\gamma^I$ is Hermitian, $e_I = 1$ and $\{K_{I},\cdot\}_{e_{I}} = \{H_{I},\cdot \}$, where $H_{I}$ is some Hermitian matrix.
	If $\gamma^I$ is skew-Hermitian, $e_I = -1$ and $\{K_{I},\cdot\}_{e_{I}} = [L_{I},\cdot ]$, where $L_{I}\,$is some skew-Hermitian matrix. 
	One can deduce from this result that the space of  possible Dirac operators $\mathcal{D}$ is isomorphic as a real vector space to some Cartesian product of copies of the spaces of $N$ by $N$ Hermitian matrices $\mathcal{H}_{N}$, and $N$ by $N$ traceless Hermitian matrices $\mathcal{H}_{N}^{0}$. In the large $N$ limit, these traceless Hermitian matrix ensembles have the same moments as their Hermitian counterparts \cite{d2022numerical}. Hence, since we are currently only interested in the large $N$ distribution of these models, we will strictly consider Hermitian matrices.
	
	With quantum gravity as a motivation, it makes sense to then consider a probability distribution on $\mathcal{D}$ called a Dirac Ensemble (or sometimes a random fuzzy geometry). The usual probability distributions of choice are of the form 
	$$\frac{1}{Z}e^{- \tr V(D)}dD,$$
	where $V$ is some polynomial with coupling constants as coefficients such that the probability distribution is well-defined.
	In \cite{khalkhali2022spectral}, Gaussian Dirac ensembles were studied extensively and found to have universal properties in the large $N$ limit. However, note that the main choice of potential in most works has been a quartic action 
	$$V(D) = g\tr D^2 + \tr D^4,$$
	since it  has been seen to already exhibit many interesting properties. In particular, quartic Dirac ensembles of this form  exhibit manifold-like behaviour near spectral phase transitions \cite{barrett2016monte,glaser2017scaling,barrett2019spectral,glaser2023computational}. If an additional coupling constant is considered in front of the quartic term, when tuned to criticality, such models have been found to have connections to the $(3,2)$ and $(5,2)$ minimal models from conformal field theory \cite{hessam2023double}. Until this work, explicit analytical progress on such models has mostly been on Dirac ensembles with only one Hermitian matrix \cite{khalkhali2020phase,khalkhali2022spectral,hessam2023double,verhoeven2023geometry}, with the notable exception of \cite{perez2022multimatrix} and bounds on moments obtained in \cite{hessam2022bootstrapping}. For more details we refer the reader to the recent review \cite{hessam2022noncommutative}.

	\subsection{Outline of main results}
	In this paper the main objects of study include the following Dirac ensemble for signatures $(2,0),(1,1)$, and $(0,2)$: 
	\begin{equation}\label{eq:measure}
		d\mu^{(p,q)} (D) =\frac{1}{Z}e^{-t_{2} \tr D^2 - t_{4}\tr D^{4}}dD.
	\end{equation}
	Its partition function is
	\begin{equation}\label{eq:partition function 1}
		Z= \int_{\mathcal{D}^{(p,q)}}e^{-t_{2} \tr D^2 - t_{4}\tr D^{4}}dD
	\end{equation}
	where $t_{2}$ and $t_{4}$ are some real coupling constants. Originally in \cite{barrett2015matrix,barrett2016monte}, a key aspect of these models is that they are convergent matrix integrals and no perturbative expansion or renormalization techniques are required to make them well-defined mathematical objects. However, in this work we shall consider both the convergent models and their perturbative counterparts. A formal matrix integral is a well-defined formal series defined by series expanding all non-Gaussian terms in the potential and then interchanging summation and integration. These are vastly different mathematical objects that historically have caused confusion, but have a deep relationship. For more details see \cite{albeverio20011,ercolani2003asymptotics,guionnet2005combinatorial,eynard2011formal}. In particular, we will show that in the large $N$ limit the loop equations for these models are the same for both the formal and convergent models, and have a unique solution. We will denote expectation with respect to a formal matrix integral with bra and ket, and with respect to a convergent matrix ensemble with $\mathbb{E}[\cdot]$.
	
	The Dirac operators on these fuzzy geometries can be expressed as
	$$D_{(2,0)} =\sigma_{3} \otimes \{A,\cdot\}+\sigma_{1} \otimes \{B,\cdot\}$$
	$$D_{(1,1)} =\sigma_{3} \otimes [A,\cdot]+\sigma_{1} \otimes \{B,\cdot\}$$
	$$D_{(0,2)} =\sigma_{3} \otimes [A,\cdot]+ \sigma_{1}\otimes [B,\cdot]$$
	where $A$ and $B$ are  $N\times N$ Hermitian matrices, and $\sigma_{1}$ and $\sigma_{3}$ are the Pauli spin matrices 
	$$\sigma_{1}= \begin{bmatrix}
		0 & 1 \\
		1 & 0
	\end{bmatrix} \qquad \qquad \sigma_{3}= \begin{bmatrix}
		1 & 0 \\
		0 & -1
	\end{bmatrix}.$$

	The commutators and anti-commutators are represented as matrices 
	$$\{A,\cdot\} = A \otimes I_{N} + I_{N}\otimes A^{T} $$
	$$[A,\cdot] = A \otimes I_{N} - I_{N}\otimes A^{T} $$
	via the isomorphism between $M_{N}(\mathbb{C})$ and $\mathbb{C}^{N}\otimes \mathbb{C}^{N}$.
	
	Expressing the action in terms of $A$ and $B$ gives us
	\begin{equation}\label{eq:action}
		\begin{aligned}
			S(D)= & 4 {t_2} \left[N\operatorname{Tr}A^2+N\tr B^2+\epsilon_1\left(\operatorname{Tr} A\right)^2+\epsilon_2\left(\operatorname{Tr} B\right)^2\right] \\
			& +4 t_{4} N\left[\operatorname{Tr} A^4 +\tr B^4+4 \operatorname{Tr} A^2 B^2-2 \operatorname{Tr} A B A B\right] \\
			& +4 t_{4} \left[4 \epsilon_1 \operatorname{Tr} A^3 \operatorname{Tr} A+4 \epsilon_2 \operatorname{Tr} B^3 \operatorname{Tr} B+3\left(\operatorname{Tr} A^2\right)^2+3\left(\operatorname{Tr} B^2\right)^2\right] \\
			& +16 t_{4} [\epsilon_1 \operatorname{Tr} A B^2 \operatorname{Tr} A+ \epsilon_2 \operatorname{Tr} B A^2 \operatorname{Tr} B ]\\
			& +8 t_{4} [\operatorname{Tr} A^2 \operatorname{Tr} B^2+2 \epsilon_1 \epsilon_2\left(\operatorname{Tr} A B\right)^2 ].
		\end{aligned}
	\end{equation}
	where the epsilons are signs that change depending on the signature of the fuzzy geometry according to Table \ref{table:signs}.
		\begin{table}[H]
			\centering
			\begin{tabular}{l c l c l c l c l}
				\hline
				KO & $\epsilon_{1}$& $\epsilon_{2}$\\
				\hline
				(2,0) & \,1 &  \,1 \\
				(1,1) &  \,1 & -1   \\
				(0,2) &  -1 &  -1 
			\end{tabular}
			\caption{Different signs in the action correspond to different KO dimension \cite{barrett2015matrix}.}
			\label{table:signs}
		\end{table}
	In this paper, all results will be to leading order in the large $N$ limit. As pointed out in \cite{khalkhali2022spectral,hessam2022bootstrapping}, many of the terms in \eqref{eq:action} do not contribute to the leading order loop equations. As such we can consider a simplified model whose action we will refer to as the effective action, but who has the exact same large $N$ behaviour as the above models:
	\begin{equation}\label{eq: effective action}
		\begin{aligned}
			S_{\text{eff}}(A,B) = & 4 t_2 N\left[\operatorname{Tr}A^2+\tr B^2\right] \\
			& +4 t_{4} N\left[\operatorname{Tr} A^4 +\tr B^4+4 \operatorname{Tr} A^2 B^2-2 \operatorname{Tr} A B A B\right] \\
			& +12t_{4} \left[\left(\operatorname{Tr} A^2\right)^2+\left(\operatorname{Tr} B^2\right)^2\right] \\
			& +8 t_{4} [\operatorname{Tr} A^2 \operatorname{Tr} B^2 ].
		\end{aligned}
	\end{equation}
	Notice that all the epsilon terms are not included. This serendipitously implies that to leading order in $N$ all the models have the same large $N$ behaviour.
	
	The model is a bi-tracial two-matrix ensemble. In random matrix theory one is generally interested in computing moments and more generally correlation functions. The \textit{Dirac moments} are defined as follows:
	\begin{equation*}
		\frac{1}{Z} \int_{\mathcal{H}_{N}^{2}} \frac{1}{N^{2}}\tr D^{\ell} e^{-S_{\text{eff}}(A,B)}dAdB,
	\end{equation*}
	for integers $\ell \geq 0$. Let $W$ belong to the set of noncommutative polynomials in two matrix variables $\mathbb{C}[A,B]$, then the (mixed) moments are defined as
	\begin{equation*}
		\frac{1}{Z} \int_{\mathcal{H}_{N}^{2}} \frac{1}{N}\tr W e^{-S_{\text{eff}}(A,B)}dAdB.
	\end{equation*}
	Note that $A$ and $B$ are symmetric in the potential, which implies the equivalence of many moments.

	Finding moments at finite $N$ is very difficult, but generally computing their limits as $N$  goes to infinity greatly simplifies calculations, provided the limit exists. Much success has been achieved in this direction dating back to Wigner \cite{wigner1967random}, and his successors \cite{brezin1978planar}. Additionally, many universal properties have been observed in the limit \cite{deift2009random}. 
	
	For unitary invariant ensembles, one can apply techniques such as the Coloumb gas/equilibrium measure approach \cite{de1995statistical, deift1999orthogonal} or in the case of a formal integral one can apply (Blobbed) Topological Recursion \cite{eynard2007invariants,borot2014formal,borot2017blobbed}. Analytic progress has been achieved for some models that lack invariance \cite{eynard1996more,eynard2003master,eynard2004genus}, but for general potentials very little is known \cite{eynard2011formal, guionnet2005combinatorial}. Our model is particularly challenging. It clearly lacks unitary invariance, techniques such as the characteristic expansion \cite{kazakov1999two}, the Harish chandra formula \cite{itzykson1980planar}, and bi-orthogonal polynomials \cite{bertola2002duality} are not applicable. Numerical studies of these particular models have been carried out and many interesting properties have been found \cite{barrett2016monte,glaser2017scaling,barrett2019spectral}. In particular in \cite{hessam2022bootstrapping} the bootstrap technique was applied to find  explicit bounds for moments of these models in the large $N$ limit. Numerical estimates for the moments were then obtained.  While an explicit formula for any moments in terms of coupling constants escaped us at the time, this paper presents such a formula.

	\begin{theorem}\label{thm:main result}
		The formal and convergent models of \eqref{eq:partition function 1} for all three signatures have the same limiting moments. In particular, for $t_{2}$ and $t_{4}$ in a sufficiently small enough neighbourhood of zero,
		$$\lim_{N \rightarrow \infty} \frac{1}{N} \mathbb{E}[\tr A^2]= 
		\frac{1}{32 t_{4}} \left({\sqrt{{t_2}^2+8
				{t_4}}}-{{t_2}}\right).$$
	\end{theorem}
	The proof is presented in Section \ref{sec:second moment}. The idea of the proof is to first consider the formal counterpart of the model and prove such a claim using Feynman graphical techniques. Then, applying results from \cite{guionnet2005combinatorial}, we can conclude that the loop equations for both formal and convergent models have a unique solution. A list of explicit formulae for higher power moments can be found in Appendix \ref{Apdx:moments}. 	We conjecture that from the second moment all other moments and Dirac moments can be computed explicitly using  Schwinger-Dyson equations.
	
	Another quantity of interest in random matrix theory is the  so called free energy,
	\begin{equation}
		F_{0} = \lim_{N \rightarrow \infty}\frac{1}{N^2}\ln Z.
	\end{equation}
	If this limit exists, as a formal series, it counts some collection of colored planar maps. This limit does indeed exist for the formal model: see Section \ref{sec:pertub}.
	\begin{theorem}\label{thm:main result 2}
 	The formal models of \eqref{eq:partition function 1} for all three signatures have the same free energy given by
		\begin{equation}
			F_{0}=\lim_{N \rightarrow \infty}\frac{1}{N^{2}} \ln Z = -\frac{1}{2} + \frac{t_{2}}{t_{2} + \sqrt{t_{2}^{2} +8 t_{4}}} + \ln \left[ \frac{\pi^2}{ 64 t_{2}^2} \left(t_{2} + \sqrt{t^{2}_{2} + 8 t_{4}} \right)\right].
		\end{equation}
	\end{theorem}

	To someone familiar with the moments of matrix models, it may appear strange why these formulae are simpler than most single matrix hermitian models. Consider, for example the moments for the $(1,0)$ quartic model \cite{hessam2023double}. The reason for these more concise expressions is that the number of maps enumerated in these models is generally smaller than its single matrix cousins. This is because there are more complicated 2-cells used in the gluing and we are restricted to gluing edges that match in color.
	
	\section{The Schwinger-Dyson equations}
	The Schwinger-Dyson equations are an infinite system of non-linear recursive equations of moments that were first discovered in \cite{migdal1983loop}. They can be derived from very simple principles but can be used to deduce many properties of matrix ensembles \cite{guionnet2019asymptotics}. Processes used to solve matrix models often rely on these equations, such as topological recursion \cite{eynard2007invariants} and bootstrapping \cite{lin2020bootstraps,kazakov2022analytic}. 
	
	\subsection{Derivation and properties}
	Let $W\in \mathbb{C}\langle A,B\rangle$. The following equality follows from Stokes' theorem 
	\begin{equation}\label{eq:NSDE}
		\sum_{i,j=1}^{N}\int \frac{\partial}{\partial A_{ij}}\tr W d\mu(A,B) = 0.
	\end{equation}
	It is important to note that this equality holds for both convergent or formal matrix integrals. In the formal case, it is applied term-wise to each Gaussian integral in the formal series.  
	
	By expanding the left-hand side using the product rule one can obtain relations between (mixed) moments. For example, suppose that $W = A^{\ell}$ for some integer $\ell\geq 0$. Then equation \eqref{eq:NSDE} becomes 
	\begin{align}\label{eq:NSDE example}
	\begin{split}
		&\sum_{k=0}^{\ell -1}\mathbb{E} [\tr A^{k} \tr A^{\ell -k -1}] =  N\mathbb{E} [  8 t_{2} \tr A^{\ell+1}
		+ 16 t_{4} \tr A^{\ell+3}] \\
		&+ \mathbb{E} [ 32 t_{4} N \tr A^{\ell+1}B^2 -16 N\tr A^{\ell+1}BAB + 48 t_{4} N\tr A^2 \tr A^{\ell+1} + 16 t_{4} N \tr B^2 \tr A^{\ell+1}]
		\end{split}
	\end{align}

	Such relations are called the \textit{ Schwinger-Dyson equations} (SDE), since, unlike the usual Schwinger-Dyson equations found in single matrix models, the matrices involved may not commute, resulting in a much more vast ocean of relations to solve. Usually, one considers the generating functions of these moments to allow complex analytic techniques to solve this infinite system \cite{eynard2016counting,guionnet2019asymptotics}. However, with this model there is no clear choice of generating functions that allow for nice closed-form expressions for the SDE. Thus, in this work we are grounded to work to the level of (mixed) moments. The authors have yet to find a formula for these SDE that is concise but also informative.

	In the large $N$ limit the SDE often simplify. In particular, the factorization property 
	$$\lim_{N \rightarrow \infty} \frac{1}{N^2}\mathbb{E} [ \tr W_{1} \tr W_{2}]  = \lim_{N \rightarrow \infty} \frac{1}{N^2}\mathbb{E} [ \tr W_{1}] \mathbb{E} [ \tr W_{2}]  $$
	is exploited when possible for $W_{1},W_{2} \in \mathbb{C}\langle A,B\rangle $. In formal Hermitian matrix models this property follows from the genus expansion \cite{eynard2016counting}. From theorem 3.1 of \cite{khalkhali2022spectral}, models such as  \eqref{eq:partition function 1} have a genus expansion and hence satisfy this property. For details on how the genus expansion implies this property see the appendix of \cite{hessam2022bootstrapping}. In the large $N$ limit we introduce the following notation for the (mixed) moments of the convergent ensemble
	$$ m_{\ell_{1},\ell_{2},...,\ell_{q}}= \lim_{N \rightarrow \infty}\frac{1}{N}\mathbb{E}[\tr A^{\ell_{1}}B^{\ell_{2}} \cdots A^{\ell_{q-1}} B^{q}], $$
	and 
	$$ m^{0}_{\ell_{1},\ell_{2},...,\ell_{q}}= \lim_{N \rightarrow \infty}\frac{1}{N}\langle \tr A^{\ell_{1}}B^{\ell_{2}} \cdots A^{\ell_{q-1}} B^{q}\rangle $$
	for the formal ensemble.
	This notation is well-defined, since the model is symmetric in $A$ and $B$.

	We are interested in the SDE in the large $N$ limit. For example, in the formal case, after normalizing equation \eqref{eq:NSDE example} and taking the limit, the result is
	\begin{align*}
		\sum_{k=0}^{\ell -1} m^{0}_{k} m^{0}_{\ell -k -1} =    8 t_{2} m^{0}_{\ell+1}
		+ 16 t_{4} m^{0}_{\ell+3}+  32 t_{4} m^{0}_{\ell+1,2} -16  m^{0}_{\ell+1,1,1,1} + 64 t_{4} m^{0}_{\ell+1}m^{0}_{2}, 
	\end{align*}
	for integer $\ell \geq 0$.
	
	For any choice of initial word, such equations can be deduced. For more examples of these equations for general words, see Appendix \ref{Apdx:NSDE}.

	\subsection{From perturbative expansion to convergent integrals } \label{sec:perturb to convergent}
	As mentioned in the introduction, our strategy is to solve the formal model corresponding to \eqref{eq:partition function}, and then use known results to relate the solution to its convergent counterpart. To do this, we will use the results in \cite{guionnet2005combinatorial} and adapt them for our bi-tracial model.
	
	Consider the following formal matrix model
	\begin{equation}
		\int_{\mathcal{H}_{N}^{2}} e^{-\tr V(A,B)} dA dB,
	\end{equation}
	where 
	\begin{equation}\label{eq: extra effective action}
		\begin{aligned}
			V(A,B) = & 4 t_2 N \left[\operatorname{Tr}A^2+\tr B^2\right] \\
			& +4 t_{4} N\left[\operatorname{Tr} A^4 +\tr B^4+4 \operatorname{Tr} A^2 B^2-2 \operatorname{Tr} A B A B\right] \\
			& +32 t_{4} [\operatorname{Tr} A^2 m^{0}_{2} + m^{0}_{2} \tr B^{2} ].
		\end{aligned}
	\end{equation}
	\begin{lemma}
		Up to the leading order in $N$, the loop equations for the model \eqref{eq: extra effective action} and the model  \eqref{eq: effective action} are the same. 
	\end{lemma}
	\begin{proof}
		
	\end{proof}

	Consider the potential $V$ as a map
	$$V : \{A_{i j}, B_{k \ell}|1 \leq i,j,k,\ell\leq N\} \rightarrow \mathbb{R}.$$
	Assume $t_{2} >0$ and $t_{4} \geq 0$. The first line of equation \eqref{eq: extra effective action} is the positive sum of convex functions so it is also convex. The second line of terms in equation \eqref{eq: extra effective action} can be expressed as
	$$ \tr (AB -BA)^2 + \tr (A^2 + B^2)^2$$
	multiplied by a positive number. Hence, it is also convex. Lastly, by the positivity of the integrand, for any $N$, $m_{2} \geq 0$ and finite, so the last line is also convex. Thus, there exists a non-empty set $U$ of coupling constants such that the action $V$ is convex. 
	
	With this above observation, one can apply theorem 1.1 from \cite{guionnet2005combinatorial} to the model defined by $V$. Since the moments and SDE's of the models \eqref{eq: extra effective action} and \eqref{eq: effective action} are the same in the limit, the result applies to the latter model, giving us the following theorem.
	\begin{theorem} \label{thm: convergent}
		There exists an $\epsilon >0$ such that, for $t_{2},t_{4}\in U \cap B_{\epsilon}(0)$ and any word $W \in \mathbb{C}\langle A, B\rangle$,
		$$\lim_{N\rightarrow \infty}\frac{1}{N} \mathbb{E}[\tr W]$$ 
		converges to the unique solution to the SDE's of the effective ensemble.
	\end{theorem}
	
	This theorem implies that in some small ball of the coupling constants near zero, the convergent and formal models coincide.

	\section{The perturbative expansion}\label{sec:pertub}
	As mentioned in the introduction, a formal matrix model is a well-defined formal generating function of Gaussian matrix integrals that is constructed by expanding all non-Gaussian terms of the potential and then interchanging the order of integration and summation. The resulting Feynman diagrams of such an integral are maps (or their dual fat graphs) \cite{brezin1978planar}. This follows from the fact that Gaussian matrix integrals can be expanded in terms of maps, which can then be organized by the genus of the associated surface. In this work, we are interested in the types of maps that come from 2-matrix integrals with bi-tracial interactions, which will be introduced in the following sections.
	
	\subsection{A primer on maps} \label{sec: primer on maps}
	We will begin by introducing some general terminology on maps. A \textit{map} of genus $g$ is a  2-cell embedding of a graph into an oriented surface of genus $g$  up to orientation-preserving homeomorphisms of the surface. In this work we are focused on maps with connected graphs of genus zero, which we will refer to as \textit{planar maps}.

	Maps can be constructed by gluing the edges of polygons in an orientation-preserving manner, i.e. no twists. The unglued edges of polygons are referred to as half-edges. A \textit{rooted} map is a map with a distinguished rooted edge. Rooted maps appear when computing moments and cumulants while unrooted maps appear when computing the partition function. In particular, cumulants and the logarithm of the partition function count connected maps. Note that maps have an associated topological invariant known as their \textit{genus}, which can be computed using Euler's formula.
	
	For our model we are interested in \textit{2-colored unstable planar maps}. A 2-colored map is a map whose half-edges have one of two assigned colors. Such colors have to match that of the other half-edge they are glued to in order to form such a map. An unstable map is a map that is glued from 2-cells whose topology corresponds to unstable Riemann surfaces with boundaries i.e. a disc or cylinder. 
	
	\begin{figure}[H]
		\centering
		\includegraphics[width=0.5\textwidth]{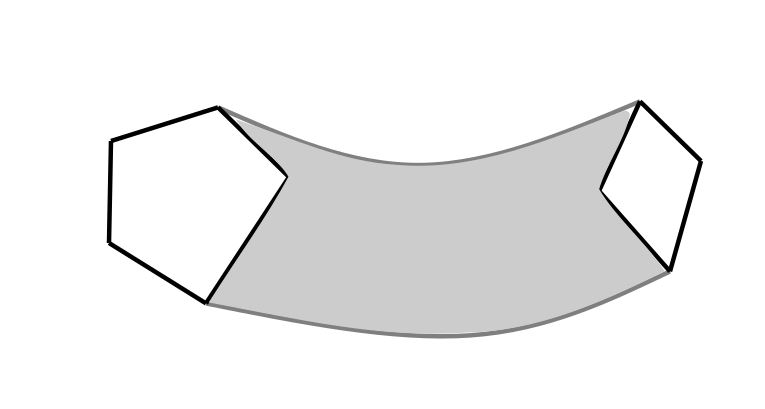}
		\caption{An example of a 2-cell with the topology of a cylinder where one boundary is a 5-gon and the other is a quadrangle}
	\end{figure}

Note that an ordinary rooted connected map glued from only polygons is planar if and only if its Euler characteristic is two. However, in unstable maps the notation of graph connectedness and map connectedness no longer coincide, so Euler's formula for genus is not always directly applicable. Unstable maps have a decomposition into graph connected components, by treating each 2-cell with the topology of a cylinder as two disconnected 2-cells with the topology of a disc. The removed part we will refer to as a \textit{branch}. Thus we can associate every unstable map with a graph, where each edge is a branch and each vertex is a graph component. One sufficient condition for an unstable map to be planar is if each graph component is planar and the associated graph described above is a tree. See \cite{khalkhali2022spectral} for more details. 
	
	The enumeration of colored maps has long been of interest in the study of formal matrix integrals, but work on unstable maps has more recently been approached in \cite{khalkhali2022spectral,hessam2023double} as well as within the more general notion of stuffed maps \cite{borot2014formal,borot2017blobbed}.  As far as the authors are aware, the enumeration of  maps with both qualities has not appeared in any works before.
	\begin{figure}[H]
		\centering
		\includegraphics[width=0.5\textwidth]{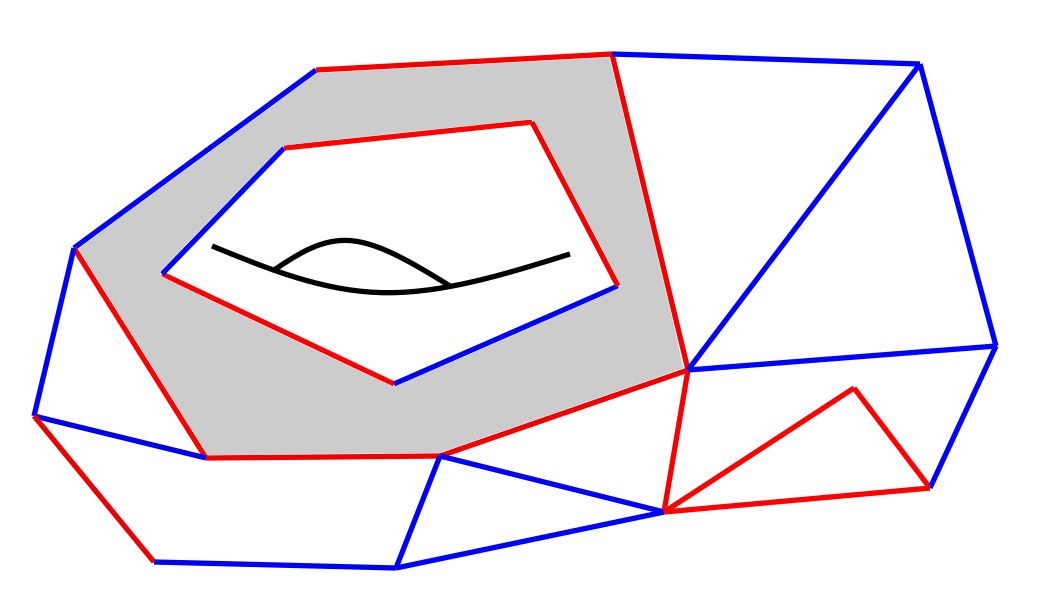}
		\caption{An example of a  2-colored unstable map of genus one.}
	\end{figure}

	\subsection{From matrix models to map enumeration}
	Let us consider the model with the action \eqref{eq: effective action} formally but we will add a redundant parameter $t$ initially, which will keep track of the number of vertices:
	\begin{equation}
		Z=\int_{\mathcal{H}_{N}^{2}}e^{- \frac{1}{t}S_{\text{eff}}(A,B)} dAdB.
	\end{equation}
	The propagators for the Gaussians are
	$$\langle A_{ij} A_{k \ell } \rangle = \frac{t}{8N}\delta_{i\ell}\delta_{jk }$$
	and 
	$$\langle B_{ij} B_{k \ell }\rangle  = \frac{t}{8N}\delta_{i\ell}\delta_{jk },$$
	where the entries of $A$ and $B$ are independent at the level of Gaussian integrals in the formal sum.
	
	Via standard techniques \cite{khalkhali2022spectral},
	the model has a  has a genus expansion, i.e. the moments can be written as 
	\begin{equation*}
		m_{W} = \sum_{g \geq 0} \left( \frac{N}{t}\right)^{1-2g}m^{g}_{W},
	\end{equation*}
	where 
	\begin{equation}\label{eq:moments genus expansion}
		m^{g}_{W} = \sum_{v=1}^{\infty} t^{v}\sum_{\Sigma \in \mathcal{UM}^{g}_{W}(v)} \frac{W(\Sigma)}{|\text{Aut}(\Sigma)|}. 
	\end{equation}
	The set $\mathcal{UM}^{g}(v)$ is the set of all  of genus $g$ 2-colored unstable maps with $v$ vertices glued from a rooted polygon whose coloring corresponds to the word $W$ and the following set of 2-cells:
	
	\begin{enumerate}
		\item A red quadrangle
		\item A blue quadrangle
		\item A quadrangle with two adjacent red edges, and two adjacent blue edges
		\item A quadrangle with two red edges whose neighbours are blue edges 
		\item A red 2-cell with the topology of a cylinder and boundaries of length two
		\item A blue 2-cell with the topology of a cylinder and boundaries of length two
		\item A 2-cell with the topology of a cylinder and boundaries of length two, one red and one blue.
	\end{enumerate}
	See Figure \ref{fig:2-cells for model} for a visualization of these 2-cells. 
	
	The realization of the correspondence of colored polygons  to cyclic words can be described as follows. The trace of a word $W$  of length $\ell$ in the alphabet formed from $A$ and $B$ has a corresponding cyclic sequence of colors of length $\ell$. This cyclic sequence of coloring is then mapped to the colors of edges of an $\ell$-gon.  See for example Figure \ref{fig:word realization }.

	\begin{figure}[H]
		\centering
		\includegraphics[width=0.5\textwidth]{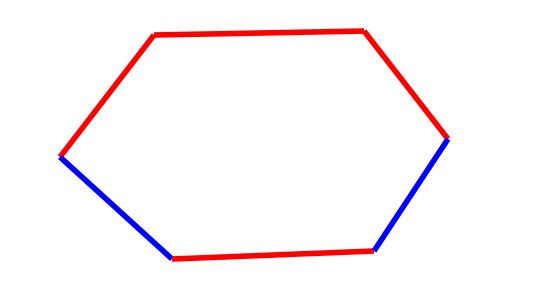}
		\caption{Moments corresponding to cyclic words $W$ in $A$ and $B$ correspond to rooted polygons with alternating colored half-edges. For example the word $AAABAB$ corresponds to the above hexagon, before choosing a root.}
		\label{fig:word realization }
	\end{figure}
	
	\begin{figure}[H]
		
		\begin{subfigure}{0.40\textwidth}
			\includegraphics[width=\textwidth]{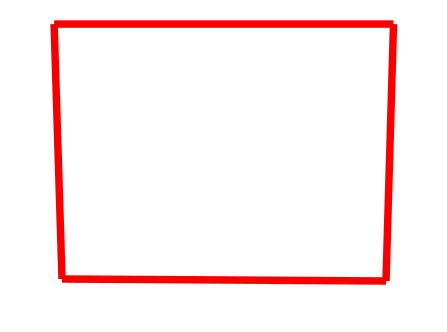}
			\caption{A blue quadrangle.}
		\end{subfigure}\hfill
		\begin{subfigure}{0.40\textwidth}
			\includegraphics[width=\textwidth]{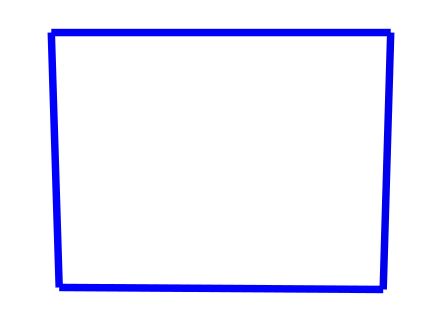}
			\caption{A red quadrangle.}
		\end{subfigure}\hfill
		\begin{subfigure}{0.40\textwidth}
			\includegraphics[width=\textwidth]{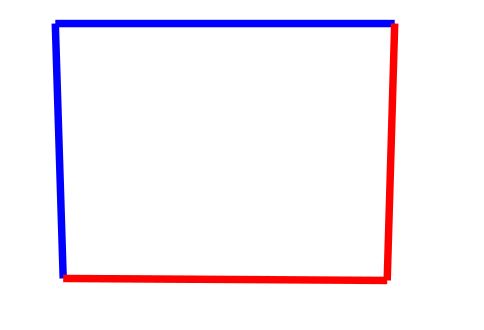}
			\caption{An adjacent colored quadrangle.}
		\end{subfigure}\hfill
		\begin{subfigure}{0.40\textwidth}
			\includegraphics[width=\textwidth]{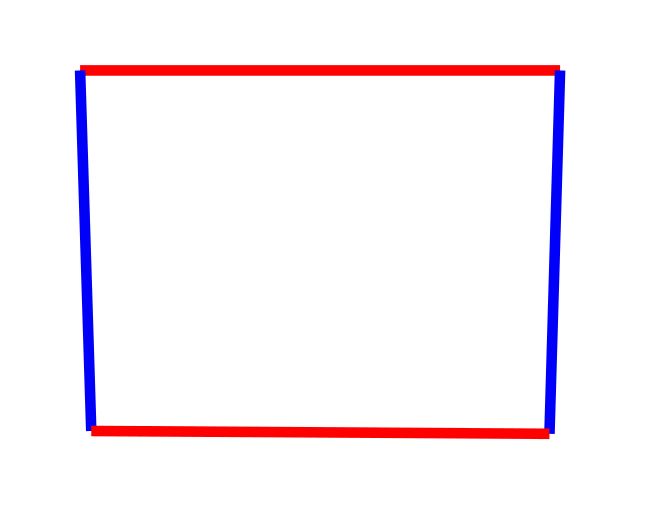}
			\caption{A chequered colored quadrangle.}
		\end{subfigure}
		
		\begin{subfigure}{0.40\textwidth}
			\vspace{10pt}
			\includegraphics[width=\textwidth]{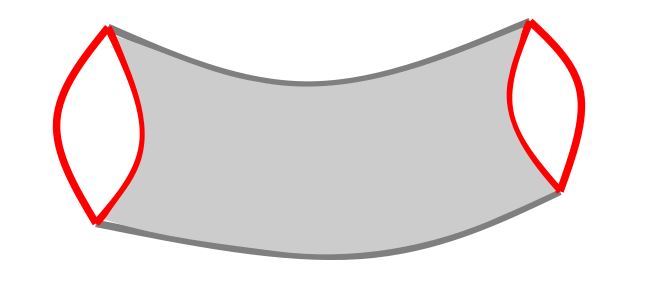}
			\caption{A red cylinder.}
		\end{subfigure}\hfill
		\begin{subfigure}{0.40\textwidth}
			\vspace{10pt}
			\includegraphics[width=\textwidth]{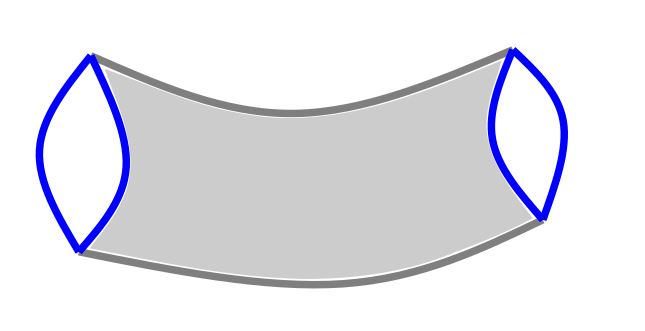}
			\caption{A blue cylinder}
		\end{subfigure}\hfill
		\begin{subfigure}{0.40\textwidth}
			\vspace{10pt}
			\includegraphics[width=\textwidth]{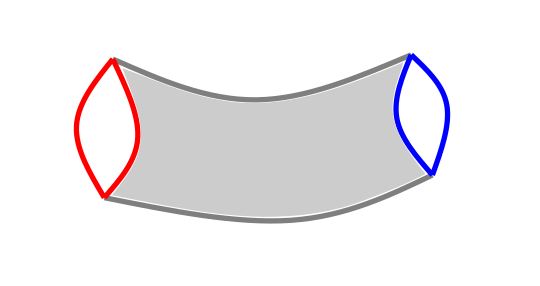}
			\caption{An opposite colored cylinder.}
		\end{subfigure}
		
		\caption{All types of 2-cells that are used in gluings of 2-colored unstable maps enumerated by the model.}
		\label{fig:2-cells for model}
		
	\end{figure}

	In equation \eqref{eq:moments genus expansion}, the Feynman weight of a map $\Sigma \in \mathcal{UM}^{g}_{W}(v)$ is given by
	\begin{equation}\label{eq:Feynman weight}
		W(\Sigma) = (16t_{4})^{n_{1}+n_{2}+n_{3}+n_{4}} (64 t_{4})^{n_{4}+n_{5}} (-16 t_{4})^{n_{7}},
	\end{equation}
	where $n_{i}(\Sigma)$, for $1\leq i\leq 7$, is the number of 2-cells corresponding to the numbers above used in the gluing of the map $\Sigma$. The coefficients in front of these weights come from a rescaling needed to construct the factor $|\text{Aut}|$ in equation \eqref{eq:moments genus expansion}. Usually, potentials are nicely normalized so that the Feynman weight is precisely the product of coupling constants. However, because there is one coupling constant in front of many terms in the effective action, this is not possible with our model. Because of this, information that helps distinguish components is lost in the final expressions, which we will find actually simplifies matters.


	\subsection{The second moment}\label{sec:second moment}
 In this section we will derive our formula for the second moment, but we must first look at another moment. Consider 
	$$m^{0}_{1,1,1,1} = \sum_{v=1}^{\infty} t^{v}\sum_{\Sigma \in \mathcal{UM}^{0}_{ABAB}(v)} \frac{\Gamma (\Sigma)}{|\text{Aut}(\Sigma)|}. $$
	Our goal is to show that this formal series is precisely zero for our model. Note that the set $\mathcal{UM}^{0}_{ABAB}(v)$ is not empty for all $v$, for an example see  Figure \ref{fig:map in set example}. Rather, we will show that the formal series is zero by showing that all positive contributions cancel with contributions from the negative sign in the Feynman weight corresponding to chequered unrooted quadrangles. 
	
	\begin{figure}[H]
		\centering
		\includegraphics[width=0.8\textwidth]{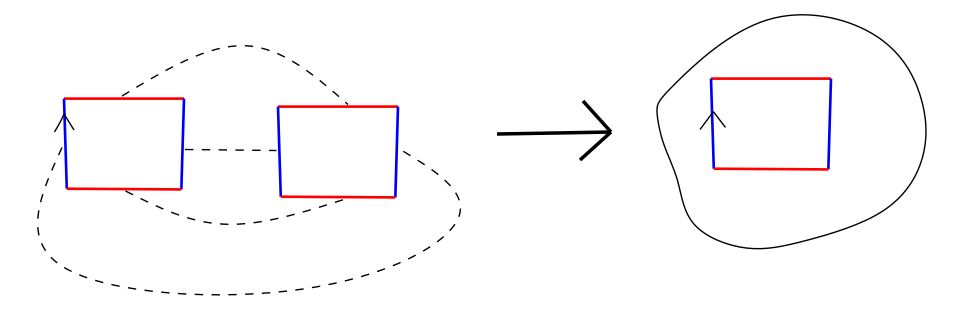}
		\caption{An example of a map in $\mathcal{UM}^{0}_{ABAB}(4)$.}
		\label{fig:map in set example}
	\end{figure}
	
	To do so, we must first study the set $\mathcal{UM}^{0}_{ABAB}(v)$. When $v=1$ or 2,the set is empty since there is no planar gluing of a rooted quadrangle with the coloring corresponding to $ABAB$. For $v>2$, the sets are not necessarily empty, but the following fact may be observed for all $v$.

	\begin{lemma}
		Any map in $\mathcal{UM}^{0}_{ABAB}(v)$ must contain at least one chequered  colored quadrangle or an opposite colored cylinder.   
	\end{lemma}\label{Lemma 1}

	\begin{proof}
		If $v=1,2$ the claim obviously holds, so let $v>2$. Consider some map $\Sigma \in \mathcal{UM}^{0}_{ABAB}(v)$ with no chequered colored quadrangles or opposite colored cylinders.  We will show such a map cannot exist. Without loss of generality, consider one of the red half-edges of the rooted face. It must be paired to another red half-edge. There are three options: a half-edge of an unrooted red quadrangle, an adjacent colored quadrangle, or a red cylinder. Note that, since the number of vertices is strictly greater than one, the edge cannot be paired with the other red half-edge of the rooted face.   In all cases, this new 2-cell must connect to another distinct red coloured 2-cell, since after considering the initial half-edge, there are an odd number of half-edges remaining. Since there are finitely many 2-cells with an even number of red half-edges used in a gluing that all need to be paired with half-edges of the same color, it must eventually connect to the other red half-edge of the rooted 2-cell by the pigeonhole principle.
		
		This above argument holds for the blue half-edges of the rooted fat vertex as well. Thus, $\Sigma$ must have at least two closed loops of colored edges that can be traced to and from the rooted face. No such map can be embedded into a sphere since this would result in these two different colored loops crossing, which is impossible without an chequered colored quadrangle.
	\end{proof}

	\begin{lemma}\label{Lemma 2}
	 For any rooted unstable colored map $\Sigma_{1}$ containing chequered colored quadrangle there is another such map with Feynman weight $W(\Sigma_{1})$. Similarly, for any rooted unstable colored map $\Sigma_{2}$ containing an opposite colored cylinder, there exists another such map with Feynman weight $W(\Sigma_{2})$ such that $W(\Sigma_{1}) = -W(\Sigma_{2})$.
	\end{lemma}
	\begin{proof}
		For $v>2$, consider a rooted unstable colored map $\Sigma_{1}$ containing an chequered colored quadrangle. We can construct a new map $\Sigma_{1}'$ containing a  an opposite colored cylinder as follows:
		\begin{enumerate}
			\item Treat one of the non-rooted chequered colored quadrangle faces in $\Sigma_{1}$ as a boundary. 
			\item Glue an opposite colored cylinder to this boundary as in Figure \ref{fig:Bijection}.
		\end{enumerate}
	
		 We claim this procedure provides us with a planar map $\Sigma_1'$. From our discussion in Section \ref{sec: primer on maps}, in order to show the gluing in Figure \ref{fig:Bijection} is planar, it  suffices to show that the left graph component of $\Sigma_{1}'$ is planar, since the right component is clearly planar and the only branch in this case forms no handles. The resulting graph component will have three fewer vertices, one fewer edge and two more faces than $\Sigma_{1}$. Thus, it has the same Euler characteristic as $\Sigma$, so also the same genus. The resulting map has the same faces except with one opposite colored cylinder instead of one chequered colored quadrangle. Recall from equation \eqref{eq:Feynman weight} that these two 2-cells contribute the same factor up to a sign in the Feynman weight. This completes the first claim.
		 
		\begin{figure}[H]
			\centering
			\includegraphics[width=0.8\textwidth]{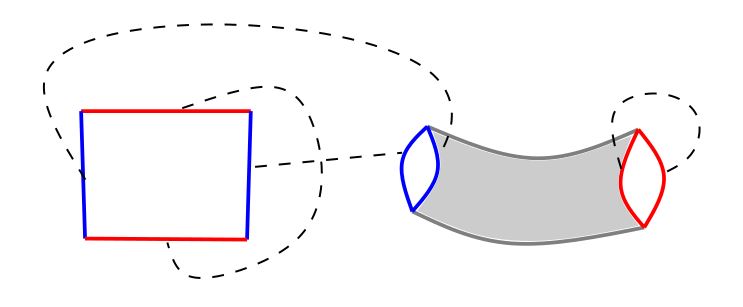}
			\caption{A planar gluing of a non-rooted adjacent colored quadrangle to an opposite colored cylinder. }
			\label{fig:Bijection}
		\end{figure}
	
		Next,  consider a map $\Sigma_{2}$ that contains   an opposite colored cylinder. Since the map is planar the branch in this cylinder must connect two distinct graph components. We then apply the following procedure:
		\begin{enumerate}
			\item If one of the 2-gons is glued to itself, do the reverse procedure of above.
			\item Otherwise:
			\begin{enumerate}
				\item  Treat both 2-gons on each graph component of $\Sigma_{2}$ as a boundary. 
			\item Glue a chequered colored quadrangle to each boundary as in Figure \ref{fig:Bijection_second}.
			\end{enumerate}
		\end{enumerate}
The resulting map $\Sigma_{2}'$ is clearly planar if the first case holds. In the second case,  $\Sigma_{2}'$ will have two more edges, one less vertex, and three more faces than either planar graph component connected by the branch in $\Sigma'$. Thus the resulting map will also be planar. In both cases, the Feynman weight $W(\Sigma_{2}) = - W(\Sigma_{2}')$.
		\begin{figure}[H]
		\centering
		\includegraphics[width=0.8\textwidth]{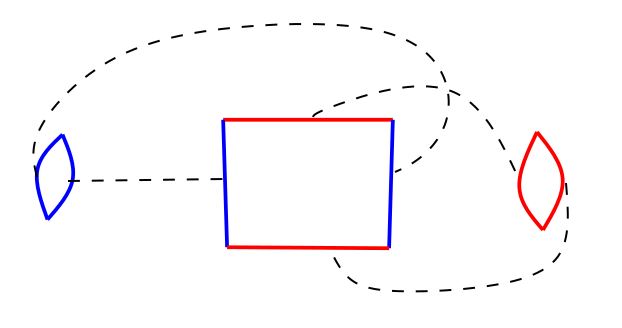}
		\caption{A planar gluing of a non-rooted adjacent colored quadrangle to two different coloured 2-gons.}
		\label{fig:Bijection_second}
	\end{figure}
	
	\end{proof}
	
	\begin{theorem}
		For the formal model with the effective action, $m^{0}_{1,1,1,1}$ is exactly zero.
	\end{theorem}
	\begin{proof}
		In the effective action, the redundant parameter $t$ that counts the number of vertices is set to one. For a discussion on why this parameter is redundant see Chapter 1.2.3 of \cite{eynard2016counting}. We also know that the set $\mathcal{UM}^{0}_{ABAB}(v)$ is empty for $v=1,2$.  It is clear then that our moment is of the form
		$$m^{0}_{1,1,1,1} = \sum_{v=3}^{\infty}\,\, \sum_{\Sigma \in \mathcal{UM}^{0}_{ABAB}(v)} \frac{\Gamma (\Sigma)}{|\text{Aut}(\Sigma)|}. $$
		We known from Lemma \ref{Lemma 1} that each map must contain at least one chequered colored quadrangle or an opposite colored cylinder. We also know from Lemma \ref{Lemma 2} that for each map with a chequered colored quadrangle there is a map with the same gluing configuration, except that the chequered quadrangle is replaced with an opposite colored cylinder and vice versa.  Additionally, note that any map with one root has a trivial automorphism group. The result is that when we collect terms of the same power, the number of terms with a positive sign will always equal the number of terms with a negative sign.
	\end{proof}
	
	With these results we may now succinctly prove the main result of this work.
	\begin{proof}[Proof of Theorem \ref{thm:main result}]
		Algebraically solving the loop equations in Appendix \ref{Apdx:NSDE} in terms of $t_{2}$, $t_{4}$, and $m^{0}_{2}$ gives the formula 
		$$m_{1,1,1,1}^{0}= \frac{2 m_{2}^2}{t_{4}} +  \frac{m_{2}\sqrt{t_{2}}}{8t_{4}^2} - \frac{1}{64 t_{4}}$$
		Using the fact that $m^{0}_{1,1,1,1}=0$, we can rearrange for $m_{2}$ in terms of $t_{2}$ and $t_{4}$. There are two roots, but we must choose the one that is not always negative for $t_{4}>0$, in order for the second moment to be positive. 
		
		Combining this with Theorem \ref{thm: convergent} gives the main result. 
		
	\end{proof}
	 Based on computations done in \cite{hessam2022bootstrapping}, we conjecture that from the second moment all other moments can be computed recursively. The proof of this conjecture at the moment seems to be a challenging combinatorial problem.

	\section{The free energy}
	\subsection{Derivation} For this particular model we can use our knowledge of the second moment to compute the leading order term of the logarithm of the partition function in the large $N$ limit, commonly referred to as the free energy \cite{chekhov2006free,chekhov2006hermitian}.

	We know from \cite{khalkhali2022spectral} that the free energy of our models has the genus expansion 
	\begin{equation}
		\ln Z= \sum_{g \geq 0} \left( \frac{N}{t}\right)^{2-2g}F_{g},
	\end{equation}
	where 
	$$F_{g} = \sum_{v=1}^{\infty} t^{v}\sum_{\Sigma \in \mathcal{UM}^{g}(v)} \frac{W(\Sigma)}{|\text{Aut}(\Sigma)|}. $$
	The set $\mathcal{UM}^{g}(v)$ is the set of all maps  of genus $g$ with $v$ vertices glued from the list of 2-cells in Figure \ref{fig:2-cells for model}. The proof of Theorem \ref{thm:main result 2} follows from a simple  computation from the formula for the fourth moment which can be found in Appendix \ref{Apdx:moments}.
	\begin{proof} [Proof of Theorem \ref{thm:main result 2}]
		It is clear that 
		\begin{equation}
			-\lim_{N \rightarrow \infty}\frac{\partial }{\partial t_{4}}\frac{1}{N^{2}}\ln Z = d_{4}.
		\end{equation}
		Since this is a formal series in $N^{-2}$, we may swap the order of the limit and differentiation,
		\begin{align*}
			-  d_{4} =\lim_{N \rightarrow \infty}\frac{\partial }{\partial t_{4}} \frac{1}{N^{2}}\ln Z &= \frac{\partial }{\partial t_{4}} \lim_{N \rightarrow \infty}\frac{1}{N^{2}}\ln Z. \\
		\end{align*}
		Integrating both sides and using the formula for the Gaussian Dirac ensemble in the large $N$ limit from Appendix \ref{App:Gaussian}, we arrive at
		\begin{align}
			F_{0} = \lim_{N \rightarrow \infty}\frac{1}{N^{2}}\ln Z &= \lim_{N \rightarrow \infty}\frac{1}{N^{2}}\ln Z|_{t_{4}=0} - \left[\int (d_{4}(s)) ds\right]_{s=t_{4}} \\
			&= - 5 \ln2 + 2 \ln \pi - 2 \ln t_{2} -\left[\int\frac{t_{2}^{2} - t_{2}\sqrt{8 s + t_{2}^2} + 4 s}{8 s^2} ds\right]_{s=t_{4}} \\
			&= -\frac{1}{2} + \frac{t_{2}}{t_{2} + \sqrt{t_{2}^{2} +8 t_{4}}} + \ln \left[ \frac{\pi^2}{ 64 t_{2}^2} \left(t_{2} + \sqrt{t^{2}_{2} + 8 t_{4}} \right)\right]
		\end{align}
	\end{proof}
	
			\subsection{Random maps and criticality}
			The free energy can be used to find critical behavior of the model, from which an asymptotic expansion can be computed. Such expansions have been shown to bridge connections to theories of 2D quantum gravity \cite{di19952d}. In \cite{hessam2023double}, several Dirac ensembles were shown to have the same critical exponents and asymptotic partition functions as various minimal models.  We would also like to emphasize that this critical behavior does not correspond to a spectral phase transition which is of interest for Dirac ensembles \cite{barrett2016monte,khalkhali2020phase}, but rather the type of critical behavior mentioned that connects matrix models to random commutative geometries. In some sense this can be thought of as a continuum limit.
			
	For a formal matrix integral with a genus expansion, its weighted map generating functions have an interpretation as a discrete probability distribution. 	For simplicity set $t_{2}=1$. Consider, for example, planar maps. If a non-trivial configuration of coupling constants are such that $F_{0}$ is finite and greater than zero, then we say such a configuration is admissible. For admissible configurations we are then able to define the discrete probability distribution over $\mathcal{UM}^{0}(v)$,
	$$\frac{1}{F_{0}}\frac{\Gamma (\Sigma)}{|\text{Aut}(\Sigma)|}.$$
	 We know that there exist admissible configurations from Theorem \ref{thm:main result 2}. 
	 
	 Usually, in matrix models, each coupling constant corresponds to a different trace term in the potential. In these cases we can compute the expectation number of 2-cells of a certain topology by differentiating the free energy. This is not the case in our potential, so a map theoretic interpretation of the second derivative is more complicated. However, it is still a quantity of interest and can roughly be thought of as a weighted expected  number of 2-cells from Figure \ref{fig:2-cells for model}, 

 \begin{equation}
	 \frac{\partial}{\partial t_{4}} 	\lim_{N \rightarrow \infty}\frac{1}{N^{2}}\ln Z = - \lim_{N \rightarrow \infty}\frac{1}{N^{2}}\langle \tr D^4 \rangle = -\frac{t_{2}^{2} - t_{2}\sqrt{8 t_{4} + t_{2}^2} + 4 t_{4}}{8 t_{4}^2}.
\end{equation}
	
With this interpretation we can see that the expected number of 2-cells diverges along the critical curve 
	\begin{equation}
		t_{4} = -\frac{1}{8} t_{2}^{2}.
	\end{equation}
 Note that the solutions of this equation are only when $t_{4}$ is less than zero. Thus, this critical behavior is only seen in the formal model and not the convergent solution.
 
 There exist formal notions of convergence of these probability distributions on maps to random metric spaces. The critical exponent of interest here is the first non-zero power, which is usually of the form $1-\gamma$. This is known as the string susceptibility exponent,  and often indicates to what random metric spaces the above one will converge to in the Gromov-Hausdorff topology. For more details we refer the reader to \cite{budd2022lessons}. 
  
 For simplicity set $t_{2}=1$. We may asymptotically expand the partition function around the critical point  $t_c=-1/8$, 
 $$1-4 \sqrt{2}
 (t_4-t_{c})^{\frac{1}{2}}+ +24 (t_4- t_c)+64 \sqrt{2} \left(t_4-t_{c}\right)^{\frac{3}{2}} + \mathcal{O}((t_4-t_{c})^{2}).$$
 This allows us to deduce that $\gamma=1/2$, which is associated with the limiting metric space known as the continuum random tree. This is also an exponent that does not appear often in random matrix models, but is common in tensor models \cite{lionni2018colored}. This may suggest that the maps enumerated here have a realization as the triangulations seen in tensor models. For comparison, the quartic type $(1,0)$ Dirac ensemble studied in \cite{hessam2023double} has a string susceptibility exponent of $-1/2$, which is associated with the Brownian map. However, future work is still needed to establish such a convergence.

	\section{Conclusions and Outlook}
	
In this work we computed the second moment and the free energy of the quartic type $(2,0)$, $(1,1)$, and $(0,2)$ Dirac ensembles in the large $N$ limit. This was done by studying properties of unstable colored maps and the associated  Schwinger-Dyson equations (SDE's). Applying the results of \cite{guionnet2005combinatorial}, we were then able to show that the solution for all moments for both the convergent and formal models up to the leading order is unique. Furthermore, we explicitly computed the first twenty moments of these models.

	These results can be compared to past numerical work. For example, the plot of the second moment in Figure 7a and 7b of \cite{glaser2017scaling}  bears a strong resemblance to the large $N$ solution presented here. This seems to indicate that convergence is rather fast since the matrix size in these simulations was rather small, between five and ten.  Additionally, in Figure \ref{fig:bootstraps}, the solution derived here can be seen to perfectly fit within the  bootstrapped bounds computed in \cite{hessam2022bootstrapping}, as expected. 
	
\begin{figure}[H]
	\centering
	\begin{subfigure}[b]{0.3\textwidth}
		\centering
		\includegraphics[width=\textwidth]{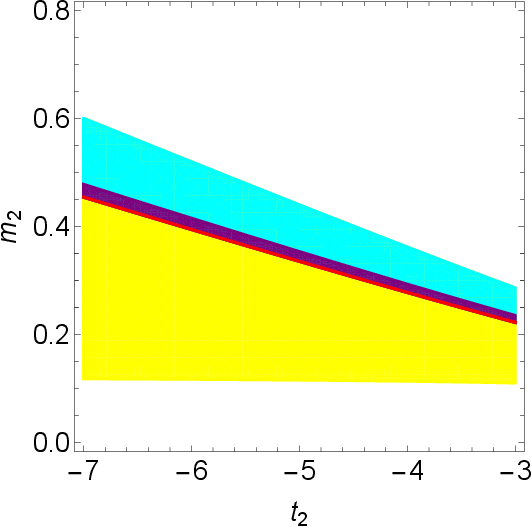}
	\end{subfigure}
	\begin{subfigure}[b]{0.3\textwidth}
		\centering
		\includegraphics[width=\textwidth]{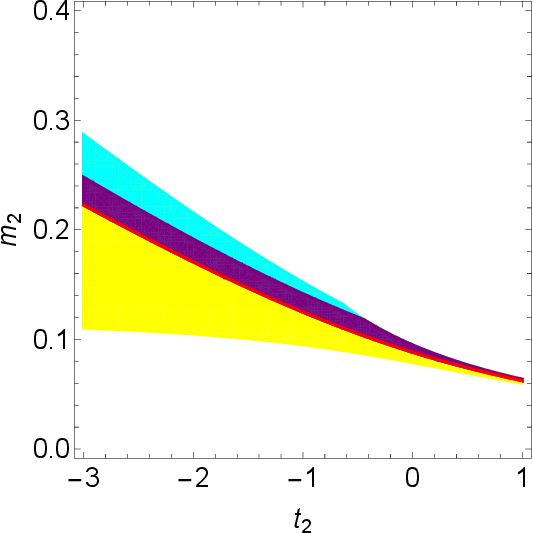}
	\end{subfigure}
\caption{A reproduction of Figures 3 and 4 from ~\protect\cite{hessam2022bootstrapping} with the second moment solution from Theorem ~\protect\ref{thm:main result} overlaid in red. In ~\protect\cite{hessam2022bootstrapping}, the models are the same except $t_{4}$ is set to one. Each color corresponds to a different region of possible solutions to the SDE's generated by considering various positivity constraints on the moments.} 
\label{fig:bootstraps}
\end{figure}

	One limitation of what we currently understand is that despite being able to compute many moments, the generating functions of both moments and Dirac moments is an enigma. It is not clear which choice of generating functions the SDE's can be written  succinctly in terms of. The types of generating functions used in studying past multi-matrix models such as the two-matrix model \cite{eynard2004genus} do not seem to be enough to close these equations. A natural candidate might be a generating function in terms of Dirac moments, but such a formulation is not known to the authors at this time. 
	
	We hope that this work will lead to a general formula for all Dirac moments, or equivalently its generating functions; as well as new techniques to study previously unsolvable multi-matrix models. In particular, ideas presented here may be useful in studying Dirac ensembles on gauge-matrix spectral triples \cite{perez2022multimatrix}. It would be also interesting to see if the recent work in \cite{parraud2023free} can be generalized to multi-trace matrix models, potentially leading to solutions to all orders of even more challenging Dirac ensembles.
		\section*{Acknowledgements}
	We would like to thank Hamed Hessam for the long discussions and hard work that preceded these results. We are also grateful for funding from the Natural Sciences and Engineering Research Council of Canada (NSERC).

	\appendix

	\section{Gaussian Dirac ensembles}\label{App:Gaussian}
	Consider the partition function for the Gaussian  type $(2,0), (1,1)$  and $(0,2)$ Dirac ensemble's partition function
	\begin{equation}\label{eq:partition function}
		Z= \int_{\mathcal{D}^{(p,q)}}e^{-t_{2}\tr D^2}dD;
	\end{equation}
	where
	$$\tr D^{2} = 4 t_2 \left[N\operatorname{Tr}A^2+\tr B^2+\epsilon_1\left(\operatorname{Tr} A\right)^2+\epsilon_2\left(\operatorname{Tr} B\right)^2\right]$$
	and $\epsilon_{1}$ and $\epsilon_{2}$ are according to Table \ref{table:signs}. We wish to compute it explicitly for any given $N$ and a strictly positive coupling constant $t_{2}$. 
	
	
	In the large $N$ limit the answer will be the same as that of the following integral since they have the same loop equations:
	\begin{align*}
		\lim_{N \rightarrow \infty} \frac{1}{N^{2}}\ln Z &= \lim_{N \rightarrow \infty} \frac{1}{N^{2}}\ln \int_{\mathcal{H}_{N}^{2}}e^{-4t_{2} N\tr (A^2 + B^{2})}dAdB\\
		&=\lim_{N \rightarrow \infty} \frac{1}{N^{2}}\ln \left(\int_{\mathcal{H}_{N}}e^{-4t_{2} N\tr A^2}dA \right)^{2}\\
		&=\lim_{N \rightarrow \infty} \frac{1}{N^{2}}\ln \left(\int_{\mathcal{H}_{N}}e^{-4t_{2} \tr A^2}dA \right)^{2}\\
		&=\lim_{N \rightarrow \infty} \frac{1}{N^{2}}\ln \left(\int_{\mathcal{H}_{N}}e^{-4t_{2} N[ \sum_{i=1}^{N} x_{i,i}^{2} + 2\sum_{i < j}(x_{i,j}^{2} + y_{i,j}^{2})] }\prod_{i} d x_{i,i} \prod_{i <j} d x_{i,j} dy_{i,j} \right)^{2}\\
		&= \lim_{N \rightarrow \infty} \frac{1}{N^{2}} \ln \left (\left(\frac{\pi}{ 4 t_{2} } \right )^{N^{2}} \left(\frac{\pi}{ 8 t_{2} } \right)^{N^{2} -N} \right)\\
		&= - 5 \ln2 + 2 \ln \pi - 2 \ln t_{2}
	\end{align*}

	\section{The limiting moments}\label{Apdx:moments}
	By the parity of the integral and the $A | B$ symmetry the normalized moments simplify as follows in the large $N$ limit.

	We list here list the first twenty unique moments by considering all words of each length up to eight.  
	
	\begin{align*}
		&m_{2} = \frac{1}{32 t_{4}} \left({\sqrt{{t_2}^2+8
				{t_4}}}-{{t_2}}\right)
	\end{align*}
	
	\begin{align*}
		&m_{4} = \frac{1}{256 t_{4}^{2}} \left(-t_{2}  \sqrt{t_{2}^2+8 t_{4}}+t_2^2+4 t_{4}\right)\\
		&m_{2,2}=\frac{1}{512 t_{4}^{2}} \left(-t_{2} \sqrt{t_{2}^2+8t_{4}}+t_{2}^2+4 t_{4}\right)\\
		&m_{1,1,1,1} = 0
	\end{align*}
	
	\begin{align*}
		&m_{6} =-\frac{19 \left(t_{2}^3-2 t_{4}^{2} \sqrt{t_{2}^2+8t_{4}}-t_{2}^2 \sqrt{t_{2}^2+8t_{4}}+6 t_{2}
			t_{4}\right)}{3276832768 t_{4}^{3}}\\
		&m_{4,2} = \frac{-t_{2}^3+2 t_{4}^{2} \sqrt{t_{2}^2+8 t_{4}}+t_{2}^2
			\sqrt{t_{2}^2+8t_{4}}-6 t_{2} t_{4}}{3276832768 t_{4}^{3}}\\
		&m_{2,1,2,1} = -\frac{3 \left(t_{2}^3-2 t_{4}^{2} \sqrt{t_{2}^2+8 t_{4}}-t_{2}^2
			\sqrt{t_{2}^2+8 t_{4}}+6 t_{2} t_{4}\right)}{3276832768 t_{4}^{3}}\\
		&m_{3,1,1,1} = -\frac{7 \left(t_{2}^3-2 t_{4}^{2} \sqrt{t_{2}^2+8t_{4}}-t_{2}^2
			\sqrt{t_{2}^2+8t_{4}}+6 t_{2} t_{4}\right)}{32768 t_{4}^{3}}
	\end{align*}
	
	\begin{align*}
		&m_{8} =\frac{11 t_{2}^4-48 t_{2} t_{4} \sqrt{t_{2}^2+8 t_{4}}+92
			t_{2}^2 t_{4}-11 t_{2}^3  \sqrt{t_{2}^2+8 t_{4}}+104
			t_{4}^2}{524288 t_{4}^4}\\
		&m_{4,1,2,1} = \frac{5 t_{2}^4-24 t_{2} t_{4} \sqrt{{t_{2}^2}+8t_{4}}+44 t_{2}^2
			t_{4}-5 t_{2}^3  \sqrt{{t_{2}^2}+8t_{4}}+56
			t_{4}^2}{524288 t_{4}^4}\\
		&m_{6,2} = \frac{15 t_{2}^4-64 t_{2} t_{4} \sqrt{t_{2}^2 +8t_{4}}+124
			t_{2}^2 t_{4}-15 t_{2}^3  \sqrt{{t_{2}^2}+8 t_{4}}+136
			t_{4}^2}{524288 t_{4}^4}\\
		&m_{2,1,1,2,1,1} =\frac{3 t_{2}^4-16 t_{2} t_{4} \sqrt{{t_{2}^2}+8t_{4}}+28 t_{2}^2
			t_{4}-3 t_{2}^3  \sqrt{{t_{2}^2}+8t_{4}}+40
			t_{4}^2}{524288 t_{4}^4} \\
		&m_{3,1,3,1} = m_{3,3,1,1} = m_{5,1,1,1} = m_{1,1,1,1,1,1,1,1}\\
		&= \frac{3 t_{2}^4-8 t_{2} t_{4} \sqrt{{t_{2}^2}+8t_{4}}+20 t_{2}^2
			t_{4}-3 t_{2}^3  \sqrt{{t_{2}^2}+8 t_{4}}+8
			t_{4}^2}{524288 t_{4}^4}\\
		&m_{2,2,2,2} = \frac{9 t_{2}^4-40 t_{2} t_{4} \sqrt{{t_{2}^2}+8t_{4}}+76 t_{2}^2
			t_{4}-9 t_{2}^3  \sqrt{{t_{2}^2}+8 t_{4}}+88
			t_{4}^2}{524288 t_{4}^4}\\
		&m_{4,4} = \frac{11 t_{2}^4-48 t_{2} t_{4}^{3/2} \sqrt{\frac{t_{2}^2}{t_{4}}+8}+92 t_{2}^2 t_{4}-11 t_{2}^3 \sqrt{t_{4}}
			\sqrt{\frac{t_{2}^2}{t_{4}}+8}+104 t_{4}^2}{524288 t_{4}^4}\\
		&m_{2,2,1,1,1,1} = \frac{t_{2}^4+4 t_{2}^2 t_{4}-t_{2}^3 \sqrt{{t_{2}^2}+8t_{4}}-8 t_{4}^2}{524288 t_{4}^4}\\
		&m_{3.2,1,2} = \frac{5 t_{2}^4-24 t_{2} t_{4} \sqrt{{t_{2}^2}+8t_{4}}+44 t_{2}^2 t_{4}-5 t_{2}^3  \sqrt{{t_{2}^2}+8t_{4}}+56
			t_{4}^2}{524288 t_{4}^4}.
	\end{align*}
	
	Note that there are several interesting relations here that are not understood or expected, such as $m_{4} = 2 m_{2,2}$ and $m_{3,1,3,1}= m_{3,3,1,1}= m_{5,1,1,1}= m_{1,1,1,1,1,1,1,1}$ . There are also very clear patterns in these moments, but no general formula is known for them at the time of writing this paper.

	The first few Dirac moments can be written as:
	\begin{align}
		d_{2} &= \frac{1}{4 t_{4}} \left(\sqrt{t_{2}^2+8 t_{4}}-t_{2}\right)\\
		d_{4} &= \frac{1}{8 t_{4}^2}\left( t_{2}^{2} - t_{2} \sqrt{t_{2}^2 + 8 t_{4}} + 4 t_{4}\right)\\
		d_{6} &= \frac{19}{256 t_{4}^{3}} \left(- t_{2}^{3} + t_{2}^{2} \sqrt{t_{2}^{2} + 8 t_4 } - 6 t_{2} t_{4} + 2 t_{4}\sqrt{t_{2}^{2} + 8 t_4 } \right)
	\end{align}

	\section{Examples of Schwinger-Dyson equations}\label{Apdx:NSDE}
	Recall that we denote limiting moments using the notation
	
	$$ m_{\ell_{1},\ell_{2},...,\ell_{q}}= \lim_{N \rightarrow \infty}\frac{1}{N}\langle \tr A^{\ell_{1}}B^{\ell_{2}} \cdots A^{\ell_{q-1}} B^{q}\rangle. $$
	Note that this notation is well-defined, since the model is symmetric in $A$ and $B$.
	
	The first twenty nine unique SDE's for our model are listed with the corresponding input word. They are ordered based on initial word input length.

	\begin{align*}
			&A: 1=8 t_{2} m_2+t_{4}(16 m_4-16 m_{1,1,1,1}+16 m_{2,2}+16 m_{2,2}+64 m_2 m_2)
	\end{align*}

	\begin{align*}
		&A^3: 2 m_2=8 t_{2} m_4+t_{4}(16 m_6-16 m_{3,1,1,1}+16 m_{4,2}+16 m_{4,2}+64 m_2 m_4)\\
		&A B^2: m_2=8 t_{2} m_{2,2}+t_{4}(16 m_{4,2}-16 m_{3,1,1,1}+16 m_{2,1,2,1}+16 m_{4,2}+64 m_2 m_{2,2})\\
		&B A B: 0=8 t_{2} m_{1,1,1,1}+t_{4}(16 m_{3,1,1,1}-16 m_{2,1,2,1}+16 m_{3,1,1,1}+16 m_{3,1,1,1}+64 m_2 m_{1,1,1,1})\\
		&B^2 A: m_2=8 t_{2} m_{2,2}+t_{4}(16 m_{4,2}-16 m_{3,1,1,1}+16 m_{4,2}+16 m_{2,1,2,1}+64 m_2 m_{2,2})\\
	\end{align*}
	
	\begin{align*}
		&A^5 : m_2^2+2 m_4 =8 t_{2} m_6+t_{4}(16 m_8-16 m_{5,1,1,1}+16 m_{6,2}+16 m_{6,2}+64 m_2 m_6) \\
		&A^3 B^2: m_{2,2}+m_2^2=8 t_{2} m_{4,2}+t_{4}(16 m_{6,2}-16 m_{3,3,1,1}+16 m_{3,2,1,2}+16 m_{4,4}+64 m_2 m_{4,2})\\
		&A B^4: m_4=8 t_{2} m_{4,2}+t_{4}(16 m_{4,4}-16 m_{5,1,1,1}+16 m_{4,1,2,1}+16 m_{6,2}+64 m_2 m_{4,2})\\
		&B A^3 B: 0=8 t_{2} m_{3,1,1,1}+t_{4}(16 m_{3,1,3,1}-16 m_{3,2,1,2}+16 m_{3,3,1,1}+16 m_{3,3,1,1}+64 m_2 m_{3,1,1,1})\\
		&B A B^3: 0=8 t_{2} m_{3,1,1,1}+t_{4}(16 m_{3,3,1,1}-16 m_{4,1,2,1}+16 m_{3,1,3,1}+16 m_{5,1,1,1}+64 m_2 m_{3,1,1,1})\\
		&B^2 A B^2: m_2^2=8 t_{2} m_{2,1,2,1}+t_{4}(16 m_{3,2,1,2}-16 m_{3,1,3,1}+16 m_{4,1,2,1}+16 m_{4,1,2,1}+64 m_2 m_{2,1,2,1})\\
		& B^3 A B: 0=8 t_{2} m_{3,1,1,1}+t_{4}(16 m_{3,1,1,3}-16 m_{4,1,2,1}+16 m_{5,1,1,1}+16 m_{3,1,3,1}+64 m_2 m_{3,1,1,1})\\
		&B^4 A: m_4=8 t_{2} m_{4,2}+t_{4}(16 m_{4,4}-16 m_{5,1,1,1}+16 m_{6,2}+16 m_{4,1,2,1}+64 m_2 m_{4,2})\\
	\end{align*}
	
	\begin{align*}
		&A^7 : 2 m_2 m_4+2 m_6=8 t_{2} m_8+t_{4}(16 m_{10}-16 m_{7,1,1,1}+16 m_{8,2}+16 m_{8,2}+64 m_2 m_8)\\
		& A^5 B^2: m_{4,2}+m_4 m_2+m_2 m_{2,2}=8 t_{2} m_{6,2}+t_{4}(16 m_{8,2}-16 m_{5,3,1,1}+16 m_{5,2,1,2}+16 m_{6,4}+64 m_2 m_{6,2})\\
		& A^3 B^4: m_2 m_4+m_{4,2}=8 t_{2} m_{4,4}+t_{4}(16 m_{6,4}-16 m_{5,1,1,3}+16 m_{4,1,2,3}+16 m_{6,4}+64 m_2 m_{4,4})\\
		& AB^6: m_6=8 t_{2} m_{6,2}+t_{4}(16 m_{6,4}-16 m_{7,1,1,1}+16 m_{6,1,2,1}+16 m_{8,2}+64 m_2 m_{6,2})\\
		& B A^5 B: 0=8 t_{2} m_{5,1,1,1}+t_{4}(16 m_{5,1,3,1}-16 m_{5,2,1,2}+16 m_{5,1,1,3}+16 m_{5,3,1,1}+64 m_2 m_{5,1,1,1})\\
		& B A^3 B^3: 0=8 t_{2} m_{3,3,1,1}+t_{4}(16 m_{3,3,3,1}-16 m_{4,1,2,3}+16 m_{3,3,3,1}+16 m_{5,1,1,3}+64 m_2 m_{3,3,1,1})\\
		&B A B^5: 0=8 t_{2} m_{5,1,1,1}+t_{4}(16 m_{5,3,1,1}-16 m_{6,1,2,1}+16 m_{5,1,3,1}+16 m_{7,1,1,1}+64 m_2 m_{5,1,1,1})\\
		&B^2 A^5 : m_{2,2} m_2+m_{4,2}+m_2 m_4=8 t_{2} m_{6,2}+t_{4}(16 m_{8,2}-16 m_{5,1,1,3}+16 m_{6,4}+16 m_{5,2,1,2}+64 m_2 m_{6,2})\\
		&B^2 A^3 B^2: 2 m_2 m_{2,2}=8 t_{2} m_{3,2,1,2}+t_{4}(16 m_{3,2,3,2}-16 m_{3,3,3,1}+16 m_{4,3,2,1}+16 m_{4,1,2,3}+64 m_2 m_{3,2,1,2})\\
		&B^2AB^4 : m_2 m_4=8 t_{2} m_{4,1,2,1}+t_{4}(16 m_{4,3,2,1}-16 m_{5,1,3,1}+16 m_{4,1,4,1}+16 m_{6,1,2,1}+64 m_2 m_{4,1,2,1})\\
		& B^3 A^3 B: 0=8 t_{2} m_{3,3,1,1}+t_{4}(16 m_{3,3,1,3}-16 m_{4,3,2,1}+16 m_{5,3,1,1}+16 m_{3,3,3,1}+64 m_2 m_{3,3,1,1})\\
		&B^3 A B^3: 0=8 t_{2} m_{3,1,3,1}+t_{4}(16 m_{3,1,3,3}-16 m_{4,1,4,1}+16 m_{5,1,3,1}+16 m_{5,1,3,1}+64 m_2 m_{3,1,3,1})\\
		& B^4 A^3: m_{4,2}+m_4 m_2=8 t_{2} m_{4,4}+t_{4}(16 m_{6,4}-16 m_{5,3,1,1}+16 m_{6,4}+16 m_{4,3,2,1}+64 m_2 m_{4,4})\\
		& B^4 AB^2: m_4 m_2=8 t_{2} m_{4,1,2,1}+t_{4}(16 m_{4,1,2,3}-16 m_{5,1,3,1}+16 m_{6,1,2,1}+16 m_{4,1,4,1}+64 m_2 m_{4,1,2,1})\\
		& B^5 A B: 0=8 t_{2} m_{5,1,1,1}+t_{4}(16 m_{5,1,1,3}-16 m_{6,1,2,1}+16 m_{7,1,1,1}+16 m_{5,1,3,1}+64 m_2 m_{5,1,1,1})\\
		& B^6 A: m_6=8 t_{2} m_{6,2}+t_{4}(16 m_{6,4}-16 m_{7,1,1,1}+16 m_{8,2}+16 m_{6,1,2,1}+64 m_2 m_{6,2})
	\end{align*}

	\section{Redundancy of coupling constants}
	In previous literature \cite{barrett2016monte,glaser2017scaling,barrett2019spectral,khalkhali2020phase,hessam2022bootstrapping}, quartic Dirac Ensembles considered have had slightly different coupling constants,
	\begin{equation}
		Z[g] = \int_{\mathcal{D}^{(p,q)}} \exp{\{-g\tr D^2 - \tr D^4\}}dD.
	\end{equation}
	We briefly address the difference in this section and show how, via a simple change of variables, we can transform the model and moments to 
	\begin{equation}
		Z[t_{2},t_{4}] = \int_{\mathcal{D}^{(2,0)}} \exp{\{-g_{2} \tr D^2 - \tr D^4\}}dD
	\end{equation}
	and its moments, allowing one to compare previous results.

	Apply the transformation $D \rightarrow D/\sqrt[4]{t_{4}}$ to $Z_{N}[t_{2},t_{4}]$ for $t_{4} \not = 0$. If $t_{4} =0$, this is precisely a Gaussian integral which can be computed explicitly for any $N$. The result is 
	\begin{equation}
		Z_{N}[t_{2},t_{4}] = \frac{1}{\sqrt[4]{t_{4}}}\int_{\mathcal{D}^{(p,q)}} \exp{\left\{-\frac{t_{2}}{\sqrt{t_{4}}} \tr D^2 - \tr D^4\right\}}dD.
	\end{equation}
	Thus,
	\begin{equation}\label{partition functions relation}
		Z_{N}[t_{2},t_{4}]=\frac{1}{\sqrt[4]{t_{4}}}Z_{N}[t_{2}/\sqrt{t_{4}}].
	\end{equation}
Define 
	\begin{equation}
		d_{\ell}[t_{2}/\sqrt{t_{4}}] =\lim_{N\rightarrow \infty}\frac{1}{N^{2}}\frac{1}{Z_{N}[t_{2}/\sqrt{t_{4}}]}  \int_{\mathcal{D}^{(p,q)}} \tr D^{\ell} \exp{\{-t_{2}/\sqrt{t_{4}} \tr D^2 + \tr D^4\}}dD.
	\end{equation}
	and 
	\begin{equation}
	d_{\ell}[t_{2}, t_{4}] =\lim_{N\rightarrow \infty}\frac{1}{N^{2}}\frac{1}{Z_{N}[t_{2}, t_{4}]}  \int_{\mathcal{D}^{(p,q)}} \tr D^{\ell} \exp{\{-t_{2} \tr D^2 + -t_{4}\tr D^4\}}dD.
	\end{equation}
	
Then, applying the same transformation  and relation \eqref{partition functions relation} to the moments, we arrive at
	\begin{align*}
		d_{\ell}[t_{2}, t_{4}] &= \lim_{N\rightarrow \infty}\frac{1}{N^{2}}\frac{1}{Z_{N}[t_{2}, t_{4}]}  \int_{\mathcal{D}^{(p,q)}} \tr D^{\ell} \exp{\{-t_{2} \tr D^2 + -t_{4}\tr D^4}\}dD\\
	&=\lim_{N\rightarrow \infty}\frac{1}{N^{2}}\frac{\sqrt{t_{4}}}{Z_{N}[t_{2}/\sqrt{t_{4}}]}  \int_{\mathcal{D}^{(p,q)}}\frac{1}{t_{4}^{\ell/4}} \tr D^{\ell} \exp{\{-t_{2}/\sqrt{t_{4}} \tr D^2 + -\tr D^4\}}\frac{1}{\sqrt{t_{4}}}dD\\
	& = t_{4}^{-\ell/4} d_{\ell}[t_{2}/\sqrt{t_{4}}].
	\end{align*}

The above relation gives a clear method to compare our results with those in the above mentioned works.

	\bibliographystyle{abbrv}
	\bibliography{references}

\end{document}